\documentclass[a4paper,12pt]{article}

\usepackage{authblk}
\usepackage{color,graphicx}
\usepackage{amssymb}
\usepackage{amsmath}
\usepackage{amsfonts}
\usepackage{amsthm}
\usepackage{enumitem}
\usepackage{fullpage}
\usepackage{pstricks}

\usepackage[T1]{fontenc}

\theoremstyle{plain}
\newtheorem{definition}{Definition}
\newtheorem{lemma}[definition]{Lemma}
\newtheorem{theorem}[definition]{Theorem}

\usepackage{clrscode3e}
\usepackage[noend, noline, ruled, linesnumbered]{algorithm2e}
\SetAlFnt{\normalsize}
\DontPrintSemicolon
\SetKwComment{Comment}{\quad\quad$\triangleright$\ }{}
\SetKwProg{Def}{def}{:}{}
\SetArgSty{textnormal}

\newcommand{\cL}{\mathcal{L}}
\newcommand{\cQ}{\mathcal{Q}}

\newcommand{\bigo}{\mathcal{O}}
\newcommand{\inE}[2]{E^{+}_{#2}(#1)}
\newcommand{\outE}[2]{E^{-}_{#2}(#1)}

\title{\textbf{Low-degree spanning trees of $2$-edge-connected graphs in linear time}\thanks{This work was supported by Intel Corporation.}}

\author[1]{Dariusz Dereniowski\thanks{The author has been partially supported by National Science Centre (Poland) grant number 2018/31/B/ST6/00820.}}
\author[2]{Janusz Dybizbański}
\author[3]{Przemysław Karpiński}
\author[2]{~\\ Michał Zakrzewski}
\author[2]{Paweł Żyliński}

\affil[1]{Gdańsk University of Technology, 80-233 Gdańsk, Poland}
\affil[2]{University of Gdańsk, 80-309 Gdańsk, Poland}
\affil[3]{Intel Technology Poland Ltd., 80-298 Gdańsk, Poland}

\date{}

\begin{document}

\maketitle

\begin{abstract}
\noindent We present a simple linear-time algorithm that finds a spanning tree $T$ of a given $2$-edge-connected graph $G$ such that each vertex $v$
of $T$ has degree at most $\lceil \frac{\deg_G(v)}{2}\rceil + 1$. 
\end{abstract}

\section{Problem statement and our result}

In 2015, Hasanvand~\cite{H15} proved that every $k$-edge-connected graph $G=(V_G,E_G)$, $k \ge 1$, has a spanning tree $T$ such that for each vertex $v \in V_T$, we have 
\[\deg_T(v) \le \left\lceil \frac{\deg_G(v)}{k}\right\rceil + 1,\] 
and this bound is tight. The particular case for $k=2$ had been already solved much earlier, independently, by several authors, see for example~\cite{BJTY03,LX98,T08}, leading to polynomial-time algorithms for determining a relevant spanning tree, in particular, to the one of (inexplicit) $\mathcal{O}(nm)$-time complexity~\cite{BJTY03} or the one of (explicit) $\mathcal{O}(m^2)$-time complexity~\cite{LX98}, where $n=|V_G|$ and $m=|E_G|$.

To the best of our knowledge, the question of the existence of a linear-time algorithm for the $k=2$ problem in $2$-edge-connected graphs remains open, and herein, we answer this question affirmatively, that is, we prove the following theorem.

\begin{theorem} 
\label{thm:algorithm}
There exists a $\mathcal{O}(m)$-time algorithm that for any $2$-edge-connected graph $G=(V_G,E_G)$ finds its spanning tree $T$ such that for each vertex $v \in V_G$, it holds
\[
\deg_T(v) \le \left\lceil \frac{\deg_G(v)}{2}\right\rceil + 1.
\]
\end{theorem}

\noindent We emphasize that   
Czumaj and Strothmann~\cite{CzS97}   proposed an $\bigo(m + n^{3/2})$-time algorithm for that
problem, $n=|V_G|$, but in the class of $2$-connected graphs (a subclass of $2$-edge-connected graphs); its randomized version achieves running time of $\bigo(m + n\log^4n)$, with high probability. Therefore, our algorithm significantly improves their result as well. We note in passing, that for an arbitrary integer $k \ge 3$, Czumaj and Strothmann~\cite{CzS97} also gave an $\mathcal{O}(n^2 \cdot k \cdot \alpha(n, n) \cdot \log n)$-time algorithm for finding a maximum degree-$\Delta_T$ spanning tree of any $k$-connected graph of maximum degree bounded by $k(\Delta_T- 2) + 2$.

\paragraph{Paper organization.} The bound in Theorem~\ref{thm:algorithm} is proven in Section~\ref{sec:degree}, the correctness is due to Lemma~\ref{lem:finalT} and the running time is estimated in Lemma~\ref{lem:time}.
The algorithm itself is listed in Section~\ref{sec:algorithm} as Algorithm~\ref{alg:mst-algorithm}.
The graph theory terminology not presented here can be found for example in~\cite{Chartrand}.

\paragraph{Background.} Conceptually, we follow the idea of~\cite{BJTY03}: we construct a strongly connected balanced orientation $D$ of a $2$-edge-connected input graph $G$ (ref.~our Lemma~\ref{lem:DFS-balanced}) and output an out-branching of $D$. The main difference is how the orientation $D$ is constructed. In our case, we use a simple edge DFS-like approach, with a single pass, while the authors in~\cite{BJTY03} exploit an iterative -- path by path -- approach.\footnote{Our approach can be easily converted to the latter approach, with the same linear running time.}

\section{Preliminaries} \label{sec:preliminaries}
Let $G=(V_G,E_G)$ be a connected graph. Let $N_G(v)$ denote the set of neighbors of a vertex $v$ in $G$, and let $E_G(v)$ be the set of edges incident to $v$ in $G$.
We say that a~list $\cL=(a_1,\ldots,a_l)$ is an \emph{edge DFS} if the following conditions hold.
\begin{enumerate}[label={\normalfont(\roman*)}]
 \item\label{it:edge-DFS1} $a_j=(u_j,v_{j})$, where $e_j=\{u_j,v_{j}\}\in E_G$ for each $j\in\{1,\ldots,l\}$.
 
 \item\label{it:edge-DFS2}
 $E_G=\{e_1,\ldots,e_l\}$.
 
 \item\label{it:edge-DFS3} 
 Define for any $v\in V_G$ and $j\in\{1,\ldots,l\}$, $X_j(v)=E_G(v) \setminus\{e_1,\ldots,e_j\}$. (Informally, the $X_j(v)$ is the set of edges incident to $v$ and not traversed by the prefix of $\cL$ of length $j$).
 Let $j\in\{1,\ldots,l-1\}$.  
\begin{itemize}
\item    If $X(v_j)\neq\emptyset$, then  $e_{j+1}\in X(v_j)$, where $u_{j+1}=v_j$.
\item If $X(v_j)=\emptyset$, then there exists $1\leq j'<j$ for which $X(v_{j'})\neq\emptyset$.
 For the maximum such $j'$ it holds $e_{j+1}\in X(v_{j'})$, where $u_{j+1}=v_{j'}$.
\end{itemize} 
\end{enumerate}
Intuitively, Condition~\ref{it:edge-DFS1} restricts the list $\cL$ to the edges of $G$, while Condition~\ref{it:edge-DFS2} says that it `covers' all edges of $G$.
Condition \ref{it:edge-DFS3} makes $\cL$ a DFS-like traversal.
In particular, if $\cL$ arrives at $v_j$ that has some incident untraversed edge(s), then one of those edges serves as $a_{j+1}$.
Otherwise, $\cL$ backtracks to the last vertex $v_{j'}$ that has untraversed edges and traverses one of them.
Condition~\ref{it:edge-DFS2} adds here that $X_l(v)=\emptyset$ for each $v\in V_G$.

We point out that we encode the traversals as a sequence of directed pairs since this notation will be more natural in our proofs.
In particular, we heavily use that each $a_i$ provides not only the next traversed edge but also encodes the direction of the traversal.
In what follows, to shorten the notation, we will write $e(a_i)$ to refer to $e_i$ in \ref{it:edge-DFS1}.

\section{The algorithm} \label{sec:algorithm}
Our algorithm is given under the name \ref{alg:mst-algorithm} below.

\begin{figure}[htb!]
\begin{center}
\begin{minipage}{.95\linewidth}
\begin{algorithm}[H]
   \SetAlgoRefName{\textup{\tt Low-Degree-ST}}
   \caption{\phantom{xxxxxxxxxxxxxxxxxxxxxxxxxxxxxxxxxxxxxxxxxxxxxxxxx} (input: a graph $G$; output: a spanning tree $T$ of~$G$)}
   \label{alg:mst-algorithm}

   Compute an edge DFS $\cL$ of $G$. \label{ln:ST:compute-L}\;
   Let $T$ be an initial tree that consists of one arbitrary vertex \label{ln:ST:init-T} \;
   \While{$V_G\neq V_T$}{
      Take any $(u,v) \in \cL$ such that $u\notin V_T$ and $v\in V_T$ \label{ln:ST:pick-edge} \;
      Add the edge $\{u,v\}$ to $T$ \label{ln:ST:extend-T}
   }
   \Return $T$
\end{algorithm}
\end{minipage}
\end{center}
\end{figure}

\section{The complexity and correctness} \label{sec:complexity}

Let $\cL=((u_1,v_1),\ldots,(u_l,v_l))$ be an edge DFS.
We say that $\cL$ \emph{crosses} $v_{i}$ when $v_i=u_{i+1}$, $i\in\{1,\ldots,l-1\}$, otherwise $\cL$ \emph{backtracks to} $u_{i+1}$, or simply \emph{backtracks} when the vertex is not important.
Whenever $\cL$ backtracks to $u_{i+1}$, $v_i$ is called \emph{final in} $\cL$.
The vertex $u_1$ is called the \emph{starting vertex}.

For the purpose of analysis, we denote by $T_j$'s all intermediate trees obtained during the execution of the algorithm.
Namely, $T_0$ is the initial tree from line~\ref{ln:ST:init-T}, and if $T_j$ is the tree $T$ 
at the beginning of an iteration of the \texttt{while} loop in Algorithm~\ref{alg:mst-algorithm}, then $T_{j+1}$ is the tree $T$ at the end of this iteration.

Note that we suggested in the above definition that $T_j$'s are trees.
This follows from the fact that whenever the algorithm adds an edge $\{u,v\}$ to $T$ in line~\ref{ln:ST:extend-T}, it checks that $u$ does not belong to $T$ yet, and $v$ is in $T$.
Hence, each $T_j$ is connected and has no cycles, which immediately results in the following lemma.

\begin{lemma} \label{lem:only-trees}
For each $j\geq 0$, $T_j$ is a tree.
\qed
\end{lemma}

Having proved that $T$ `grows correctly', it remains to argue that it eventually becomes a spanning tree, i.e., the last tree $T_k$ in the sequence satisfies $V_{T_k}=V_G$.
We say that $(X,Y)$ is a \emph{partition} of $G$ if $X\cup Y=V_G$ and $X\cap Y=\emptyset$.
We have the following straightforward property.
\begin{lemma} \label{lem:partition}
If $(X,Y)$ is a partition of a $2$-edge-connected $G$, then there exist at least two different edges, each having one endpoint in $X$ and one in $Y$.
\qed
\end{lemma}

\begin{lemma} \label{lem:finalT}
The last tree obtained during the execution of Algorithm \ref{alg:mst-algorithm} is a spanning tree of $G$.
\end{lemma}
\begin{proof}
Let $n=|V_G|$.
We prove by induction on $j\in\{0,\ldots,n-2\}$ that there exists an item $(u,v)$ in $\cL$ such that $u\notin V_{T_j}$ and $v\in V_{T_j}$, and hence $T_{j+1}$ is computed by the algorithm, provided that the algorithm obtained $T_j$.
The base case of $j=0$ and the inductive step for $j>0$ are identical, hence pick any $j\in\{0,\ldots,n-2\}$.
Consider the partition $(X,Y)=(V_{T_j},V_G \setminus V_{T_j})$ of $G$.

Take $(u_i,v_i)\in\cL$ such that $u_i\in X$ and $v_i\in Y$.
If no such item in $\cL$ exists, then each edge $\{x,y\}$ with $x\in X$ and $y\in Y$ provides an item $(y,x)\in\cL$, and then Algorithm \ref{alg:mst-algorithm} can extend $T_j$ to obtain $T_{j+1}$ using this edge.
So suppose that the item $(u_i,v_i)$ exists in $\cL$ and let $i$ be the minimum such index.
Let $\widetilde{Y}\subseteq Y$ be such that $v_i\in \widetilde{Y}$ and $\widetilde{Y}$ is the maximum subset that induces a connected component in $G$.
By Lemma~\ref{lem:partition} applied for partition $(V_G\setminus\widetilde{Y},\widetilde{Y})$, there exists an edge $\{x,y\}$ different than $\{u_i,v_i\}$ such that $x\in V_G\setminus\widetilde{Y}$ and $y\in \widetilde{Y}$.
By the maximality of $\widetilde{Y}$, we must have $x\in X$.

Find the minimum $i'>i$ such that the edge $e=\{u_{i'},v_{i'}\}$ gives the item $(u_{i'},v_{i'})$ in $\cL$, such that one endpoint of $e$ is in $X$ and the other -- in $\widetilde{Y}$.
Due to the existence of $\{x,y\}$ such an $i'$ exists, although it may be the case that $e\neq\{x,y\}$ if many edges connecting $X$ and $\widetilde{Y}$ exist.
If $u_{i'}\in Y$ and $v_{i'}\in X$, then Algorithm \ref{alg:mst-algorithm} can add $e$ to $T_j$ in line~\ref{ln:ST:extend-T}, which proves that $T_{j+1}$ is obtained and completes the proof.
However, we note that we have not proved that $e$ is the edge added to $T_j$ to obtain $T_{j+1}$ but we only have that $e$ is a suitable candidate for this extension.

The case when $u_{i'}\in X$ and $v_{i'}\in \widetilde{Y}$ remains to be discussed. 
We prove by contradiction that this is not possible.
Consider the part of $\cL$ that follows the item $(u_i,v_i)$.
A simple iterative argument gives that, for this part, it holds that if the last vertex visited is in $\widetilde{Y}$, which initially holds because we start with $v_i\in\widetilde{Y}$, then we have three possibilities.
Either an item $(a,b)$ follows in $\cL$, where $a\in\widetilde{Y}$ and $b\in X$.
Or, an item $(a,b)$ follows in $\cL$, where $\{a,b\}\subseteq\widetilde{Y}$ and $\cL$ crosses $a$.
Or, $\cL$ backtracks to a vertex $a$.

In the first case, where informally speaking $\cL$ `crosses' from $\widetilde{Y}$ back to $X$, we have $\{a,b\}=\{u_{i'},v_{i'}\}$ by the minimality of $i'$.
But this means that $a=u_{i'}$ and $b=v_{i'}$, thus contradicting $u_{i'}\in X$ and $v_{i'}\in\widetilde{Y}$.
In this case, the induction takes us to the finish of the proof of the lemma.
For the second case, $b\in\widetilde{Y}$, which takes us to the next iteration of our argument.
In the last case $\cL$ backtracks.
Note that the existence of the edge $\{x,y\}$ implies that $\widetilde{Y}$ has at least one vertex that has an incident edge $e'$ not traversed by $\cL$ till this point.
Thus, by the definition of the edge DFS, i.e., Condition~\ref{it:edge-DFS3} with constraints for backtracking, $\cL$ backtracks to a vertex that belongs to $\widetilde{Y}$, which also takes us to the next step of the iterative reasoning.
Thus, this induction eventually finishes with a contradiction, as required.
\end{proof}

\begin{lemma} \label{lem:time}
Given any $2$-edge-connected graph $G=(V_G,E_G)$ as an input, the complexity of Algorithm \ref{alg:mst-algorithm} is $\bigo(|E_G|)$.
\end{lemma}
\begin{proof}
The complexity of constructing an edge DFS $\cL$ is $\bigo(|E_G|)$.
To claim the linear running time it is enough to argue that the execution of line~\ref{ln:ST:pick-edge} can be done in constant time using a data structure whose total processing time is $\bigo(|E_G|)$.
This data structure is a queue $\cQ$ that is used in the following way.
It is initialized in line~\ref{ln:ST:init-T} where the tree~$T$ is set to contain a single vertex $x$.
The $\cQ$ is populated with all pairs from $\cL$ having $x$ as one of the endpoints.
Then, $\cQ$ is used during the execution of line~\ref{ln:ST:pick-edge} as follows.
First, repeat the following: remove the first element $(u,v)$ from $\cQ$ and if $(u,v)$ satisfies the condition in line~\ref{ln:ST:pick-edge}, then the desired pair is found.
Otherwise, dequeue $(u,v)$ and repeat with the next item in $\cQ$.
Finally, $\cQ$ is updated during the execution of line~\ref{ln:ST:extend-T} by enqueueing all pairs from $\cL$ having $v$ as one of the endpoints; denote all such pairs by $Z(v)$.

Note that enqueueing to $\cQ$ all items with one endpoint equal to a vertex $x$ for each vertex $x\in V_G$ occurs exactly once for each $x$.
Thus exactly $2|E_G|$ items are added to $\cQ$ in total.
Moreover, the operation of enqueueing them to $\cQ$ can be done in time $\bigo(|E_G(x)|)$ using pre-processing: once $\cL$ is computed, iterate over the elements of $\cL$ and for each item add it to $Z(v)$ if the item has $v$ as an endpoint.
\end{proof}

\section{Vertex degrees in $T$} \label{sec:degree}

For a vertex $u$ and an edge DFS $\cL=(a_1,\ldots,a_l)$, denote
\[\inE{u}{\cL} = \left\{ e(a_i) \colon a_i=(v_i,u) \right\},\]
\[\outE{u}{\cL} = \left\{ e(a_i) \colon a_i=(u,v_{i+1}) \right\}.\]
Each edge in $\inE{u}{\cL}$ is called \emph{incoming} and each edge in $\outE{u}{\cL}$ is \emph{outgoing}.
\begin{lemma} \label{lem:DFS-balanced}
For each $\cL$ and each vertex $u$ that is visited by $\cL$, $|\inE{u}{\cL}|\leq|\outE{u}{\cL}|+1$ if $\deg_G(u)$ is odd, and $|\inE{u}{\cL}|\leq|\outE{u}{\cL}|$ if $\deg_G(u)$ is even.
\end{lemma}
\begin{proof}
Let $\cL=(a_1,\ldots,a_l)$ and first assume that $u$ is not the starting vertex.
Whenever $\cL$ has an item $a_i=(u_i,v_{i})$ such that not all edges incident to $v_{i}$ have been traversed yet, i.e., $E_G(v_{i})\setminus\{e(a_1),\ldots,e(a_i)\}\neq\emptyset$, then $a_{i+1}=(v_{i},v_{i+1})$ due to Condition~\ref{it:edge-DFS3} in the definition of the edge DFS ($\cL$ crosses $v_i$).
Hence, $e(a_{i+1})\in\outE{v_{i}}{\cL}$.
This means that the above element in $\inE{v_{i}}{\cL}$ is uniquely paired with the corresponding element in $\outE{v_{i}}{\cL}$.
The other case is when $E_G(v_{i})\setminus\{e(a_1),\ldots,e(a_i)\}=\emptyset$.
Then, $\cL$ backtracks to $u_{i+1}$ according to Condition~\ref{it:edge-DFS3}.
This, however, may happen only once for each vertex $v_i$.

Hence, if the degree of a vertex $u$ is even, then either each element in $\inE{u}{\cL}$ is uniquely paired with an element in $\outE{u}{\cL}$, which completes the proof, 
or one element in $\inE{u}{\cL}$ is not paired, which occurs when $u$ is final in $\cL$.
Then, it follows from the parity of $\deg_G(u)$ that there is at least one edge in $\outE{u}{\cL}$ without a pair in $\inE{u}{\cL}$, which also implies the required bound.

The reasoning when $\deg_G(u)$ is odd is analogous, except for the argument that when $u$ is final, it may be the case that each element in $\inE{u}{\cL}$ is in fact paired with an element in $\outE{u}{\cL}$, which gives the worst case of $|\inE{u}{\cL}|=|\outE{u}{\cL}|+1$.

The remaining cases are when $u$ is the starting vertex for $\cL$ or when it is final in $\cL$.
We skip those as being analogous.
\end{proof}

Now, for a vertex $u \in V_T$, let us analyze its degree $\deg_T(u)$ in the final tree  $T=(V_T,E_T)$ returned by Algorithm~\ref{alg:mst-algorithm}.
First, we argue that only one outgoing edge, among those in $E_G(u)$, can belong to $E_T$.
Indeed, for an outgoing edge $\{u,v\}\in\outE{u}{\cL}$ to be added to $T$ it must hold $(u,v)\in\cL$ (see line~\ref{ln:ST:extend-T} in Algorithm~\ref{alg:mst-algorithm}).
However, this may occur only once because, after the first such event, $u\in V_T$, and subsequent outgoing edges do not result in adding another edge to $T$ due to the condition $u\notin V_T$ in line~\ref{ln:ST:pick-edge}.
Having in mind the worst case when all incoming edges incident to $u$ are added to $T$, we obtain:
\begin{equation} \label{eq:degT}
\deg_T(u)\leq |\inE{u}{\cL}|+1.
\end{equation}
If the degree of $u$ is even, then by Lemma~\ref{lem:DFS-balanced},
\[\deg_T(u) \leq \frac{1}{2}(|\inE{u}{\cL}|+|\outE{u}{\cL}|)+1 = \deg_G(u)/2+1.\]
If the degree of $u$ is odd, then again by~\eqref{eq:degT} and Lemma~\ref{lem:DFS-balanced},
\[\deg_T(u) \leq \frac{|\inE{u}{\cL}|}{2}+\frac{|\outE{u}{\cL}|+1}{2}+1.\]
Since for an odd $\deg_G(u)$ we have
\[\frac{|\inE{u}{\cL}|+|\outE{u}{\cL}|+1}{2} = \frac{\deg_G(u)+1}{2} = \left\lceil  \deg_G(u)/2 \right\rceil,\]
we have proved Theorem~\ref{thm:algorithm}.

\bibliographystyle{plain}
\bibliography{refs}

\end{document}